\documentclass[a4paper,12pt]{article}
\usepackage{amsmath}
\usepackage{amsthm}
\usepackage{amsfonts}
\usepackage{amssymb}
\usepackage{mathrsfs}
\usepackage{graphicx}
\usepackage{color}
\usepackage[colorlinks=true, allcolors=blue]{hyperref}

\newtheorem{theorem}{Theorem}
\newtheorem{lemma}{Lemma}

\usepackage[english]{babel}

\numberwithin{equation}{section}

\hypersetup{ pdftitle={Draft},
pdfauthor={Mirda Prisma Wijayanto, Emir Syahreza Fadhilla, Fiki Taufik Akbar, Bobby Aka Gunara}, 
pdfkeywords={Global Esistence, Einstein Klein Gordon}, 
bookmarksnumbered, pdfstartview={FitH}, urlcolor=blue, }

\begin{document}
\pagestyle{plain}




\title{\LARGE\textbf{Global existence of classical static solutions of four dimensional Einstein-Klein-Gordon system}}

\author{Mirda Prisma Wijayanto, Emir Syahreza Fadhilla,\\ Fiki Taufik Akbar, Bobby Eka Gunara\footnote{Corresponding author}\\ 
\\
 \textit{\small Theoretical High Energy Physics Research Group, }\\
\textit{\small Faculty of Mathematics and Natural Sciences,}\\
\textit{\small Institut Teknologi Bandung}\\
\textit{\small Jl. Ganesha no. 10 Bandung, Indonesia, 40132}
\\  \\
\\
\small email:mirda.prisma.wijayanto@students.itb.ac.id, emirsyahreza@students.itb.ac.id,\\ \small ftakbar@itb.ac.id, bobby@itb.ac.id}

\date{\today}

\maketitle




\begin{abstract}
In this paper, we prove the global existence of classical static solutions of Einstein's gravitational theory coupled to a real scalar field where spacetime admits spherically symmetry. The equations of motions can then be reduced into a single first-order integro-differential equation. First, we obtain the decay estimates of the solutions. Then, to prove the global existence, we use the contraction mapping theorem in the appropriate function spaces.
\end{abstract}

\section{Introduction}
\label{sec:intro}
One of the most interesting problems in general relativity is the study of gravitational waves which bring some information for many astronomical phenomena such as black holes collision, oscillation of the black hole horizon, the fusion of the two neutron stars in a binary pulsar system, etc. However, detecting gravitational waves require very precise devices and detectors. Apart from the problem regarding the difficulties of observing gravitational waves, another attempt to study gravitational waves is to create a mathematical model where gravity is coupled with matter fields, leading to a dynamical theory of fields that lives on a curved space-time that is also evolving in time. Such theories are expected to provide a model for gravitational waves coming from a distribution of matter as its source. One of the simple ways to do so is by introducing a minimally coupled term of a scalar (spin-0) field, also known as Klein-Gordon field, in the Einstein-Hilbert Lagrangian \cite{Carrol}. The resulting equation of motion for the fields is the Einstein equation with energy-momentum tensor coming from Klein-Gordon Lagrangian. As such, an equivalent way to construct the Cauchy problem for this model is by incorporating scalar fields into Einstein's equations with an additional equation from the divergence-less property of the scalar field's energy-momentum tensor. 

The Cauchy problem of the Einstein equations coupled with the scalar field equation was extensively studied for the past few decades. This problem was initiated by Christodolou \cite{Chris1} which considers the Einstein equations coupled with the massless scalar field as the gravity-inducing matter field in which spherically symmetric solutions are proven to exist for small initial data (and for large initial data in \cite{Chris2}). The massless Einstein-scalar system considered in the works mentioned above is special because it is proven that self-similar solutions of this system can develop a singularity in finite time for some specified initial conditions \cite{Christodoulou:1994}. Thus, it is one of the examples of naked singularity formations within Einstein's general relativity. However, it is later found that there are instabilities in this naked singularity formation, in agreement with the problem of ``cosmic censorship'' conjecture in this model that is addressed in \cite{Christodoulou:1999}.

Another important study of the coupled Einstein and scalar system is done by Malec \cite{Malec} where the Cauchy data is given on the space-like hyper-surface, as well as by examining the local existence of the Cauchy problem outside the outgoing null hyper-surfaces. Furthermore, the initial-value problem of spherically symmetric solutions to the coupled Einstein and nonlinear scalar system has been considered by Chae \cite{Chae} in four dimensions. The nonlinearity of the Einstein-scalar system considered by Chae comes from the potential term in the Klein-Gordon model that is chosen to be proportional to the scalar field to the power of any positive real number greater than or equal to three. This result opens up a new possible approach to the study of the behavior of a dynamical Einstein-scalar system with arbitrary potential terms. There are also several works of the Einstein-scalar system including \cite{Nout,Reiris,Luk,Costa} and the references therein.

In the present study, we study the global initial value problem for the Einstein-scalar system. In general relativity, this problem is devoted to answering the question of what is the condition of the initial data for the Einstein equation that provides the geodesically complete spacetime which is asymptotic to the Minkowski spacetime in the future infinity. We organize this paper as follows: In section 2, we briefly review the construction of the Einstein-scalar system in four dimensions. We start by introducing the coordinate system of our model. Then, we construct the Einstein equation and the single first-order evolution equation as the main problem. Section 3 is devoted to showing that the solution has decay properties in the $k^{th}-$order. In section 4, we employ the contraction mapping principle in the appropriate Banach space. In the last section, we give the main theorem of the global existence of a classical solution.

\section{The Einstein-Scalar Equation}
\subsection{The Coordinate System}
\label{sec:EKG}
Let us consider a Lorentzian manifold locally diffeomorphic to $\mathbb{R}^{1+3}$. The group $SO(3)$ acts as an isometry and the group orbits are the metric spacelike $2$-spheres in which the invariants of the group form a timelike curve in space-time. We have a time-like half-plane that passes through each point on $2$-spheres, with the boundary in the central world line, which is orthogonal to the group orbits and intersects each group orbit at a single point.

Let us denote $A$ the area of the group orbits. We introduce the radial coordinate $r$ satisfies $A=4\pi r^2$, such that the central world line is at $r=0$, and the world lines $r=r_0$ (constant) in each half-plane are all timelike. We also define a timelike coordinate $u$ with the property that $u$ is constant on the future light cone of each point on the central world line. On a world line $r=r_o$, $u$ tends to the proper time as $r_0\rightarrow \infty$. Hence, the metric of our spacetime can be written down in the form \cite{Chris1}
\begin{eqnarray}\label{metric}
\mathrm{d}s^2 = -e^{2F(u,r)}\mathrm{d}u^2 - 2e^{F(u,r)+G(u,r)}\mathrm{d}u\mathrm{d}r + r^2 \mathrm{d}\theta^2 + r^2\sin^2\theta \mathrm{d}\varphi^2,
\end{eqnarray}
where $F$ and $G$ are functions of $u$ and $r$, and $F\rightarrow 0$ as $r\rightarrow \infty$ at each $u$. We also require that $G\rightarrow 0$ as $r\rightarrow \infty$ at each $u$, since we only consider spacetimes that are asymptotically flat at future null infinity. This coincides with the cosmic censorship conjecture which states that there exists a solution of the Einstein equations for a given arbitrary asymptotically flat initial data which is a globally hyperbolic spacetime possessing a complete future null infinity \cite{Chris1}. Furthermore, since we only consider the future of some initial future light cone with vertex at the center, we provide the range of the coordinates $u$ and $r$ such that $0\leq u<\infty$ and $0\leq r<\infty$ respectively, where $u=0$ denotes the initial future light cone. Hence, our coordinate system is called the Bondi coordinate system.

\subsection{The Einstein Field Equation}
We write down the nonzero components of the Ricci tensor related to the metric (\ref{metric})
\begin{eqnarray}
R_{00}&=& -e^{F-G}\left[\partial_1\partial_0 (F+G) +\frac{2}{r}\partial_0 G\right]\nonumber\\
&&+ e^{2(F-G)}\left[(\partial_1F)^2+\partial_1(\partial_1 F)-\partial_1F\partial_1 G +\frac{2}{r}\partial_1 F \right]\label{R00}\\
R_{01}&=&R_{10}= -\partial_1\partial_0(F+G)+e^{F-G}\left[\partial_1(F-G)\partial_1 F + \partial_1(\partial_1 F)+\frac{2}{r}\partial_1 F\right]\label{R01}\\
R_{11}&=&\frac{2}{r}\partial_1(F+G)\label{R11}\\
R_{22}&=&1-e^{-2G}\left[r\partial_1(F-G)+1\right]\label{R22}\\
R_{33}&=&R_{22}\sin^2\theta\label{R33},
\end{eqnarray}
and the scalar curvature
\begin{eqnarray}\label{Ricci scalar}
R&=&\frac{2}{r^2}+2e^{-(F+G)}\partial_1\partial_0(F+G)\nonumber\\
&&+e^{-2G}\left[-2(\partial_1 F)^2+2(\partial_1G)(\partial_1 F)-2\partial_1(\partial_1 F)-\frac{4}{r}\partial_1(F-G)-\frac{2}{r^2}\right]
\end{eqnarray}
respectively, where we have denoted $\partial_0=\frac{\partial}{\partial u}$, and $\partial_1=\frac{\partial}{\partial r}$.  

The action of our system has the form
\begin{eqnarray}\label{action}
S=\int d^4x \sqrt{-g}\left\{\frac{1}{2}\partial_\mu\phi\partial^\mu\phi - V(\phi)\right\},
\end{eqnarray}
where $g$ is the determinant of the metric (\ref{metric}), and $\phi$ is the real scalar field.  Let us denote the scalar potential $V(\phi)$ fulfils the estimate
\begin{equation}\label{estimasi Potensial}
|V(\phi)|+\left|\frac{\partial V(\phi)}{\partial \phi}\right||\phi|+\left|\frac{\partial^2V(\phi)}{\partial\phi^2}\right||\phi|^2\leq K_0 |\phi|^{p+1},
\end{equation}
where $K_0\geq 0$ and $p\in[3,\infty)$, $\forall\phi\in\mathbb{R}$.

Now, we write the Einstein equations
\begin{eqnarray}\label{Einstein}
R_{\mu\nu}-\frac{1}{2}g_{\mu\nu}R=8\pi T_{\mu\nu},
\end{eqnarray}
with the energy-momentum tensor
\begin{eqnarray}\label{Energy-Momentum Tensor}
T_{\mu\nu}=\partial_\mu\phi\partial_\nu \phi-\frac{1}{2}g_{\mu\nu}g^{\alpha\beta}\partial_\alpha\phi\partial_\beta\phi + g_{\mu\nu}V(\phi),
\end{eqnarray}
where we have denoted $T$ as the trace of $T_{\mu\nu}$. Using (\ref{R11}), we obtain the $\{11\}$ component of (\ref{Einstein}) as
\begin{eqnarray}\label{rr}
\frac{\partial}{\partial r}(F+G)=4\pi r \left(\frac{\partial \phi}{\partial r}\right)^2.
\end{eqnarray}
The solution of (\ref{rr}) fulfills the asymptotic condition $F+G\rightarrow 0$ as $r\rightarrow \infty$ such that
\begin{eqnarray}\label{F+G}
F+G = -4\pi \int_{r}^{\infty} s\left(\frac{\partial\phi}{\partial s}\right)^2\mathrm{d}s.
\end{eqnarray}
From $(\ref{R22})$ and $(\ref{R33})$, we write down the $\{22\}$ (or $\{33\}$) components of (\ref{Einstein}) as
\begin{eqnarray}\label{ij}
\frac{\partial}{\partial r}(F-G)=-\frac{1}{r}\left(-e^{2G}+1\right)+8\pi  r e^{2G}V(\phi).
\end{eqnarray}
We put the form of the $\{00\}$ and $\{01\}$ components of  (\ref{Einstein}) in Appendix \ref{Appendix Einstein}. As we see from (\ref{rr}), (\ref{ij}), (\ref{00}), and (\ref{01}), there is no term containing the second derivative with respect to $u$ and $r$ on $F$ and $G$. Thus, the metric functions $F$ and $G$ are no longer dynamical fields, but they simply become constraints. In the rest of the paper we will use (\ref{rr}) and (\ref{ij}) for our analysis rather than (\ref{00}), and (\ref{01}) since it is easier and already common in the standard literature \cite{Chris1, Chris2}.

\subsection{The Scalar Equation of Motions}
From the action (\ref{action}), we obtain the equation of motion
\begin{eqnarray}\label{scalar}
\Box \phi = -\frac{\partial V(\phi)}{\partial \phi},
\end{eqnarray}
where $\Box$ is defined as the d'Alembert operator with respect to the metric (\ref{metric}) as follows 
\begin{eqnarray}
\Box\phi = -2 e^{-(F+G)}\left[\frac{\partial^2\phi}{\partial u \partial r}+\frac{1}{r}\frac{\partial\phi}{\partial u}\right]+ e^{-2G}\left[\frac{\partial^2\phi}{\partial r^2}+\frac{\partial\phi}{\partial r}\left(\frac{2}{r}+\frac{\partial F}{\partial r}-\frac{\partial G}{\partial r}\right)\right].
\end{eqnarray}

Let us introduce a new function
\begin{eqnarray}
h=\frac{\partial}{\partial r}(r\phi).
\end{eqnarray}
From the definition of the mean value $\bar{f}(u,r):=\frac{1}{r}\int_{0}^{r}f(u,s)\mathrm{d}s$,
we have
\begin{eqnarray}
\phi=\bar{h},~~~\frac{\partial}{\partial r}\phi = \frac{h-\bar{h}}{r}.
\end{eqnarray}
Setting
\begin{eqnarray} 
g:=e^{F+G},~~~\tilde{g}:= e^{F-G},
\end{eqnarray}
we rewrite (\ref{F+G}) as
\begin{eqnarray}\label{g2}
g=\exp\left[-4\pi\int_{r}^{\infty} \frac{(h-\bar{h})^2}{s}\mathrm{d}s \right].
\end{eqnarray}
Integrating (\ref{ij}), we have
\begin{eqnarray} \label{g tilde}
\tilde{g}=\bar{g}+\frac{8\pi}{r}\int_{0}^{r}s^2 gV(\bar{h})\mathrm{d}s.
\end{eqnarray}

We introduce a new operator
\begin{eqnarray}
\mathcal{D}=\frac{\partial}{\partial u} - \frac{\tilde{g}}{2}\frac{\partial}{\partial r},
\end{eqnarray}
that describes the derivative along the incoming light rays parametrized by $u$. Thus, we can write (\ref{scalar}) as the single first-order integro-differential equation \footnote{Eq.(\ref{Dh}) is the complete and the correct first-order integro-differential equation compared to Eq.(2.13) in \cite{Chae} if we take $V=-\frac{1}{p+1}|\phi|^{p+1}$.}
\begin{eqnarray}\label{Dh}
\mathcal{D}h = \frac{1}{2r}(h-\bar{h})\left(g-\tilde{g}\right)+4\pi g r (h-\bar{h}) V(\bar{h})+\frac{gr}{2}\frac{\partial V(\bar{h})}{\partial\bar{h}}.
\end{eqnarray}
This provides the nonlinear evolution equation with respect to $h$.

\section{Decay Estimates}
Suppose we have an initial data $h(0,r)$. Setting $k\in[3,\infty)$. Then, we define the function space
\begin{align}\label{space X}
X=\{h(.,.)\in C^1([0,\infty)\times [0,\infty))~|~\|h\|_X<\infty \},
\end{align} 
with the norm
\begin{align}\label{norm x}
\|h\|_X :=\sup_{u\geq 0} \sup_{r\geq 0} \left\{(1+r+u)^{k-1}|h(u,r)|+(1+r+u)^k \left|\frac{\partial h}{\partial r}(u,r)\right|\right\}.
\end{align}
We also define
\begin{align}
X_0 = \{h(.)\in C^1([0,\infty))~|~\|h\|_{X_0}<\infty \},
\end{align}
with the norm
\begin{align}\label{norm x0}
\|h\|_{X_0}:=\sup_{r\geq 0}\left\{(1+r)^{k-1}|h(r)|+(1+r)^k\left|\frac{\partial h}{\partial r}(r) \right| \right\}.
\end{align}
\begin{lemma} \label{Lemma 1}
	Let us denote $\|h\|_X=x$ and $\|h(0,.)\|_{X_0}=d$. Setting $k\in[3,\infty)$ and $p\in[k,\infty)$. Given the initial data $h(0,r)\in C^1[0,\infty)$ such that $h(0,r)=O(r^{-{(k-1)}})$ and $\frac{\partial h}{\partial r}(0,r)=O(r^{-k})$.  Then, the solution of (\ref{Dh}) has the decay properties:
	\begin{align} \label{decay1}
	|h(u,r)|\leq \frac{C(d+x^3+x^p+x^{p+2})\exp\left[C(x^2+x^{p+1})\right]}{(1+r+u)^{k-1}},
	\end{align}
	and
	\begin{align} \label{decay2}
	\left|\frac{\partial h}{\partial r}(u,r) \right|\leq& \frac{C(d+x^3+x^{k+2}+x^p+x^{p+2}+x^{p+4})(1+x^2+x^{k+1}+x^{p+1}+x^{p+3})}{(1+r+u)^k}\nonumber\\
	&\times \exp\left[C(x^2+x^{k+1}+x^{p+1})\right] ,
	\end{align}
	for any constant $C > 0$.
\end{lemma}
\begin{proof}
	We consider the mapping $h\mapsto \mathcal{F}(h)$ as the solution of the first-order linear partial differential equation
	\begin{align}\label{DF4 Prisma}
	\mathcal{D} \mathcal{F} = \frac{1}{2r}(\mathcal{F}-\bar{h})\left(g-\tilde{g}\right)+4\pi g r (\mathcal{F}-\bar{h})V(\bar{h})+\frac{gr}{2}\frac{\partial V(\bar{h})}{\partial\bar{h}},
	\end{align}
	with the initial condition
	\begin{align}\label{kondisi awal}
	\mathcal{F}(h)(0,r)=h(0,r).
	\end{align}
	
	Let $r(u)=\chi(u;r_0)$ be the solution of the ordinary differential equation
	\begin{align} \label{PDB}
	\frac{\mathrm{d}r}{\mathrm{d}u}=-\frac{1}{2}\tilde{g}(u,r),~~~r(0)=r_0.
	\end{align}
	We introduce the characteristic $r_1=\chi(u_1;r_0)$. Then, from (\ref{PDB}) we obtain
	\begin{align}\label{kondisi awal r}
	r_1 = r_0 - \frac{1}{2}\int_{0}^{u_1} \tilde{g}(u,\chi(u;r_0))~\mathrm{d}u.
	\end{align}
	We can represent (\ref{DF4 Prisma}) as an integral equation
	\begin{align}\label{Fcurl Prisma}
	\mathcal{F}(u_1,r_1)=&h(0,r_0)\exp\left\{\int_{0}^{u_1} \left[\frac{1}{2r}\left(g-\tilde{g}\right)+4\pi g  rV(\bar{h})\right]_\chi\mathrm{d}u\right\}\nonumber\\
	&+\int_{0}^{u_1} \exp\left\{\int_{u}^{u_1} \left[\frac{1}{2r}\left(g-\tilde{g}\right)+4\pi g rV(\bar{h})\right]_\chi\mathrm{d}u'\right\}[f]_\chi\mathrm{d}u
	\end{align}
	where we have denoted
	\begin{align}\label{f prisma}
	f=-\left[\frac{1}{2r}\left(g-\tilde{g}\right)+4\pi g  rV(\bar{h})\right]\bar{h}+\frac{gr}{2}\frac{\partial V(\bar{h})}{\partial\bar{h}}.
	\end{align}
	
	Let us denote $\|h\|_X=x$. From the mean value and the definition of the norm (\ref{norm x}) we have
	\begin{align}\label{hbar}
	|\bar{h}|\leq \frac{1}{r}\int_{0}^{r}|h(u,s)|\mathrm{d}s	\leq \frac{x}{(k-2)}\frac{r^{k-3}}{(1+u)^{k-2}(1+r+u)^{k-2}}.
	\end{align}
	Since $|\bar{h}|$ is always positive, we have $k\in[3,\infty)$. Then, we estimate
	\begin{align}
	|h(u,r)-h(u,r')|\leq\int_{r'}^{r}\left|\frac{\partial h}{\partial s}(u,s)\right|~\mathrm{d}s\leq 	\frac{x}{k-1}\left[(1+r'+u)^{-k+1}-(1+r+u)^{-k+1}\right].
	\end{align}
	From the above estimate, we obtain
	\begin{align}\label{h-h}
	|(h-\bar{h})(u,r)|\leq \frac{1}{r}\int_{0}^{r}|h(u,r)-h(u,r')|~\mathrm{d}r'\leq\frac{xr}{(k-1)(1+u)^{k-2}(1+r+u)^{k-1}},
	\end{align}
	and
	\begin{align}\label{int h-hbar}
	\int_{0}^{\infty}\frac{|h-\bar{h}|^2}{r}~\mathrm{d}r\leq\frac{x^2 }{(k-1)^2(12-14k+4k^2)(1+u)^{4(k-2)}}.
	\end{align}
	From (\ref{g2}), we get
	\begin{align}
	g(u,0)\geq \exp \left[-\frac{4\pi x^2 }{(k-1)^2(12-14k+4k^2)}\right],
	\end{align}
	and
	\begin{align}\label{g(u,0)}
	\bar{g}(u,0)\geq\exp \left[-\frac{4\pi x^2 }{(k-1)^2(12-14k+4k^2)}\right].
	\end{align}
	We obtain
	\begin{align}\label{g-gbar}
	|(g-\bar{g})(u,r)|\leq\frac{\pi x^2 r^2}{3k(k-1)^2(k-2)^2(1+u)^{-3+2k}(1+r+u)^3}.
	\end{align}
	Now, we estimate
	\begin{align}\label{estimate1}
	\int_{0}^{r}gs^2V(\bar{h})\mathrm{d}s\leq K_0\int_{0}^{r}gs^2|\bar{h}|^{p+1}\mathrm{d}s\leq\frac{K_0x^{p+1}r^3}{3(1+u)^{4(k-2)}(1+r+u)^3}.
	\end{align}
	Thus,
	\begin{align} \label{g-gtilde}
	\left|g-\tilde{g}\right|\leq|g-\bar{g}| + \frac{8\pi}{r}\int_{0}^{r}gs^2V(\bar{h})\mathrm{d}s\leq\frac{C(x^2+x^{p+1})r^2}{(1+u)^{-3+2k}(1+r+u)^3}.
	\end{align}
	Combining (\ref{estimasi Potensial}), (\ref{hbar}), and (\ref{g-gtilde}) we obtain (see Appendix \ref{Appendix1})
	\begin{align} \label{f3 Prisma}
	|f|\leq\frac{C(x^3+x^p+x^{p+2})}{(1+u)^{3(k-2)}(1+r+u)^2}.
	\end{align}
		
	Now, we estimate
	\begin{align}\label{g tilde2}
	\tilde{g}\geq\bar{g}(u,0)+\frac{8\pi}{r}\int_{0}^{r}s^2gV(\bar{h})\mathrm{d}s \geq\exp \left[-\frac{4\pi x^2 }{(k-1)^2(12-14k+4k^2)}\right]+\frac{8\pi K_0 x^{p+1}}{3}.
	\end{align}
	Let $x_1$ be the positive solution of the equation
	\begin{align} \label{x1}
	\exp \left[-\frac{4\pi x^2 }{(k-1)^2(12-14k+4k^2)}\right]+\frac{8\pi K_0 x^{p+1}}{3}=0.
	\end{align}
	We define the new function
	\begin{align}\label{k}
	\kappa=\kappa(x)=\exp \left[-\frac{4\pi x^2 }{(k-1)^2(12-14k+4k^2)}\right]+\frac{8\pi K_0 x^{p+1}}{3},
	\end{align}
	fulfills $0<\kappa\leq 1$ for all $x\in[0,x_1)$. Using the characteristic $r(u)=\chi(u;r_0)$, from (\ref{g tilde2}) and (\ref{k}) we obtain
	\begin{align}\label{ru}
	r(u)=r_1 + \frac{1}{2}\int_{u}^{u_1}\tilde{g}(u',r(u'))\mathrm{d}u'\geq r_1 +\frac{1}{2}\kappa(u_1-u),
	\end{align}
	and
	\begin{align}\label{ru2}
	1+r(u)+u\geq 1 + \frac{u}{2} + r_1 + \frac{\kappa}{2}(u_1 - u)\geq \kappa\left(1+\frac{u_1}{2}+r_1\right)\geq\frac{\kappa}{2}(1+r_1+u_1).
	\end{align}
	From (\ref{kondisi awal r}), we obtain
	\begin{align}
	r_0 = r_1 + \frac{1}{2}\int_{0}^{u_1} \tilde{g}(u,\chi(u;r_0))~\mathrm{d}u\geq r_1 +\frac{1}{2}\kappa u_1.
	\end{align}
	Let us denote $\|h(0,.)\|_{X_0}=d$. Thus,
	\begin{align}\label{h0r0}
	|h(0,r_0)|\leq \frac{\|h(0,)\|_{X_0}}{(1+r_0)^{k-1}}\leq  \frac{d}{(1+r_1+\frac{1}{2}\kappa u_1)^{k-1}}\leq \frac{2^{k-1}d}{\kappa^{k-1}(1+r_1+u_1)^{k-1}}.
	\end{align}
	
	Let us introduce the integral formula
	\begin{align}\label{integral}
	\int_{0}^{u_1}\left[\frac{r^s}{(1+u)^t(1+r+u)^q}\right]_\chi\mathrm{d}u\leq& \int_{0}^{u_1}\left[\frac{1}{(1+u)^t(1+r+u)^{q-s}}\right]_\chi\mathrm{d}u\nonumber\\
	\leq& \frac{1}{\kappa^m(1+r_1+u_1)^m}\int_{0}^{\infty}\frac{\mathrm{d}u}{(1+u)^{q-s+t-m}}\nonumber\\
	=&\frac{2^m}{(q-s+t-m-1)\kappa^m(1+r_1+u_1)^m},
	\end{align}
	with $q-s+t-m>1$, and $q,s,t,m\in\mathbb{R}$. Hence, we estimate
	\begin{align}\label{int g2}
	\int_{0}^{u_1} \left[\left|\frac{1}{2r}\left(g-\tilde{g}\right)\right|+4\pi g r|V(\bar{h})|\right]_\chi~\mathrm{d}u\leq C(x^2+x^{p+1}).
	\end{align}
	
	Combining (\ref{f3 Prisma}), (\ref{h0r0}), and (\ref{int g2}), we write the estimate of (\ref{Fcurl Prisma}) as follows:
	\begin{align}
	\left|\mathcal{F}(u_1,r_1)\right|\leq&\left|h(0,r_0)\right|\exp\left\{\int_{0}^{u_1} \left[\left|\frac{1}{2r}\left(g-\tilde{g}\right)\right|+4\pi g r|V(\bar{h})|\right]_\chi~\mathrm{d}u\right\}\nonumber\\
	&+\int_{0}^{u_1} \exp\left\{\int_{u}^{u_1} \left[\left|\frac{1}{2r}\left(g-\tilde{g}\right)\right|+4\pi g r|V(\bar{h})|\right]_\chi~\mathrm{d}u'\right\}[f]_\chi~ \mathrm{d}u\nonumber\\
	\leq&\frac{2^{k-1}d}{\kappa^{k-1}(1+r_1+u_1)^{k-1}}\exp \left[C(x^2+x^{p+1})\right]\nonumber\\
	&+\int_{0}^{u_1}\exp\left[C(x^2+x^{p+1})\right]\frac{C(x^3+x^p+x^{p+2})}{(1+u)^{3(k-2)}(1+r+u)^2}\mathrm{d}u.
	\end{align}
	Using the integral formulation (\ref{integral}), we obtain
	\begin{align}\label{F Prisma}
	|\mathcal{F}(u_1,r_1)|\leq \frac{C(d+x^3+x^p+x^{p+2})\exp\left[C(x^2+x^{p+1})\right]}{\kappa^{k-1}(1+r_1+u_1)^{k-1}}.
	\end{align}
	We obtain the decay estimate (\ref{decay1}).
	
	Now, let us define
	\begin{align}
	\mathcal{G}(u,r)=\frac{\partial \mathcal{F}}{\partial r}(u,r),
	\end{align}
	with the initial condition
	\begin{align}
	\mathcal{G}(0,r_0)=\frac{\partial h}{\partial r}(0,r_0).
	\end{align}
	Differentiating (\ref{Fcurl Prisma}) with respect to $r$ we obtain 
	\begin{align}\label{linear G}
	\mathcal{D}\mathcal{G}=&\left[\frac{1}{2}\frac{\partial \tilde{g}}{\partial r}+\frac{1}{2r} \left(g-\tilde{g}\right)+4\pi g rV(\bar{h})\right]\mathcal{G}\nonumber\\
	& + \left[\frac{1}{2r}\frac{\partial}{\partial r}\left(g-\tilde{g}\right)-\frac{1}{2r^2}\left(g-\tilde{g}\right)+\frac{\partial g}{\partial r}4\pi rV(\bar{h})+4\pi g V(\bar{h})+4\pi g r \frac{\partial V(\bar{h})}{\partial\bar{h}} \frac{\partial\bar{h}}{\partial r}\right]\left(\mathcal{F}-\bar{h}\right)\nonumber\\
	& -\left[\frac{1}{2r} \left(g-\tilde{g}\right)+4\pi g rV(\bar{h}) -\frac{gr}{2}\frac{\partial^2 V(\bar{h})}{\partial\bar{h}^2}\right]\frac{\partial \bar{h}}{\partial r}+\left[\frac{\partial g}{\partial r}\frac{r}{2}+\frac{g}{2}\right]\frac{\partial V(\bar{h})}{\partial\bar{h}}
	\end{align}
	Using the characteristics as previously, we can solve the equation (\ref{linear G}) for $\mathcal{G}(u_1,r_1)$ as follows:
	\begin{align}\label{G1r1}
	\mathcal{G}(u_1,r_1)=&\frac{\partial h}{\partial r}(0,r_0)\exp\left\{\int_{0}^{u_1}\left[\frac{1}{2}\frac{\partial \tilde{g}}{\partial r}+\frac{1}{2r} \left(g-\tilde{g}\right)+4\pi g r V(\bar{h})\right]_\chi\mathrm{d}u \right\}\nonumber\\
	&+\int_{0}^{u_1}\exp\left\{\int_{u}^{u_1}\left[\frac{1}{2}\frac{\partial \tilde{g}}{\partial r}+\frac{1}{2r} \left(g-\tilde{g}\right)+4\pi g r V(\bar{h})\right]_\chi\mathrm{d}u' \right\}[f_1]_\chi\mathrm{d}u,
	\end{align}
	where we have denoted
	\begin{align} \label{f1 Prisma}
f_1 =& \left[\frac{1}{2r}\frac{\partial}{\partial r}\left(g-\tilde{g}\right)-\frac{1}{2r^2}\left(g-\tilde{g}\right)+\frac{\partial g}{\partial r}4\pi rV(\bar{h})+4\pi gV(\bar{h})+4\pi g r \frac{\partial V(\bar{h})}{\partial\bar{h}}\frac{\partial\bar{h}}{\partial r}\right]\left(\mathcal{F}-\bar{h}\right)\nonumber\\
& -\left[\frac{1}{2r} \left(g-\tilde{g}\right)+4\pi g rV(\bar{h}) -\frac{gr}{2}\frac{\partial^2 V(\bar{h})}{\partial\bar{h}^2}\right]\frac{\partial \bar{h}}{\partial r}+\left[\frac{\partial g}{\partial r}\frac{r}{2}+\frac{g}{2}\right]\frac{\partial V(\bar{h})}{\partial\bar{h}}.
	\end{align}
	
	Differentiating (\ref{g tilde}) with respect to $r$, we get
	\begin{align}
	\frac{\partial\tilde{g}}{\partial r}=\frac{\partial\bar{g}}{\partial r}-\frac{8\pi}{r^2}\int_{0}^{r}gs^2V(\bar{h})\mathrm{d}s+8\pi grV(\bar{h}).
	\end{align}
	We define
	$\frac{\partial\bar{g}}{\partial r}=\frac{g-\bar{g}}{r}$, such that
	\begin{align}
	\left|\frac{\partial\tilde{g}}{\partial r}\right|\leq \frac{|g-\bar{g}|}{r}+\frac{8\pi}{r^2}\int_{0}^{r}gs^2|V(\bar{h})|\mathrm{d}s+8\pi gr|V(\bar{h})|.
	\end{align}
	From (\ref{hbar}), (\ref{g-gbar}), and (\ref{estimate1}), we obtain
	\begin{align}\label{gtilde}
	\left|\frac{\partial\tilde{g}}{\partial r}\right|\leq \frac{C(x^2+x^{k+1}+x^{p+1})r}{(1+u)^{-3+2k}(1+r+u)^3}.
	\end{align}
	In view of (\ref{g tilde}) and (\ref{h-h}), the above yields
	\begin{align}\label{gr}
	\left|\frac{\partial g}{\partial r}\right|\leq \frac{4\pi x^2r }{(k-1)^2(1+u)^{2(k-2)}(1+r+u)^{2(k-1)}}.
	\end{align}
	Combining (\ref{estimasi Potensial}), (\ref{hbar}), (\ref{gtilde}), and (\ref{gr}), we have (see Appendix \ref{Appendix2})
	\begin{align}\label{f1 koreksi}
	|f_1|\leq& \frac{C(x^2+x^{k+1}+x^{p+1}+x^{p+3})}{(1+u)^{2(k-2)}(1+r+u)^3}|\mathcal{F}|+\frac{C(x^3+x^{k+2}+x^{p+2}+x^{p+4})}{(1+u)^{3(k-2)}(1+r+u)^4}\nonumber\\
	&+\frac{C(x^3+x^{p}+x^{p+2})r}{(1+u)^{3(k-2)}(1+r+u)^{k+1}}.
	\end{align}
	
	Then, we obtain the estimate
	\begin{align}\label{h0r02}
	\left|\frac{\partial h}{\partial r}(0,r_0) \right|\leq \frac{\|h_0\|_{X_0}}{(1+r_0)^k}\leq\frac{d}{(1+r_1+\frac{1}{2}\kappa u_1)^k}\leq \frac{2^kd}{\kappa^k(1+r_1+u_1)^k}.
	\end{align}
	Using the integral formulation (\ref{integral}), we obtain
	\begin{align}\label{exp Prisma}
	\int_{0}^{u_1}\left[\frac{1}{2}\left|\frac{\partial\tilde{g}}{\partial r}\right|+\frac{1}{2r}|g-\tilde{g}|+4\pi gr|V(\bar{h})|\right]_\chi\mathrm{d}u\leq C\left(x^2 + x^4 + x^{p+1} \right).
	\end{align}
	
	In view of (\ref{f1 koreksi}), (\ref{h0r02}), and (\ref{exp Prisma}), we can represent the estimate (\ref{G1r1}) as follows:
	\begin{align}
	\mathcal{G}(u_1,r_1)\leq&\left|\frac{\partial h}{\partial r}(0,r_0)\right|\exp\left\{\int_{0}^{u_1}\left[\frac{1}{2}\left|\frac{\partial\tilde{g}}{\partial r}\right|+\frac{1}{2r}|g-\tilde{g}|+4\pi gr|V(\bar{h})|\right]_\chi\mathrm{d}u\right\}\nonumber\\
	&+\int_{0}^{u_1}\exp\left\{\int_{u}^{u_1}\left[\frac{1}{2}\left|\frac{\partial\tilde{g}}{\partial r}\right|+\frac{1}{2r}|g-\tilde{g}|+4\pi gr|V(\bar{h})|\right]_\chi\mathrm{d}u' \right\}\left|[f_1]_\chi\right|\mathrm{d}u\nonumber\\
	\leq&\frac{2^kd}{\kappa^k(1+r_1+u_1)^k}\exp\left[C(x^2+x^4+x^{p+1})\right]+\int_{0}^{u_1}\exp\left[C(x^2+x^4+x^{p+1})\right]\nonumber\\
	&\times \left[\frac{C(x^2+x^{k+1}+x^{p+1}+x^{p+3})}{(1+u)^{2(k-2)}(1+r+u)^3}|\mathcal{F}|+\frac{C(x^3+x^{k+2}+x^{p+2}+x^{p+4})}{(1+u)^{3(k-2)}(1+r+u)^4}\right.\nonumber\\
	&\left.+\frac{C(x^3+x^{p}+x^{p+2})r}{(1+u)^{3(k-2)}(1+r+u)^{k+1}}\right]\mathrm{d}u.
	\end{align}
	Using the integral formulation (\ref{integral}) and the estimate (\ref{F Prisma}), we have
	\begin{align}\label{G Prisma}
	\left|\mathcal{G}(u_1,r_1)\right|\leq&\frac{C(d+x^3+x^{k+2}+x^p+x^{p+2}+x^{p+4})(1+x^2+x^{k+1}+x^{p+1}+x^{p+3})}{\kappa^k(1+r_1+u_1)^k}\nonumber\\
	\leq&\exp\left[C(x^2+x^4+x^{p+1})\right].
	\end{align}
	We obtain the decay estimate (\ref{decay2}).
	This is the end of the proof.
\end{proof}
\section{Contraction Mapping}
We define the function space $Y$ containing $X$ as follows
\begin{align} \label{space Y}
Y=\{h(.,.)\in C^1([0,\infty))\times[0,\infty))~|~h(0,r)=h_0(r),~\|h\|_Y<\infty\},
\end{align}
with the norm
\begin{align}\label{norm y}
\|h\|_Y = \sup_{u\geq 0}\sup_{r\geq 0}\left\{(1+r+u)^{k-1}|h(u,r)| \right\}.
\end{align}

\begin{lemma}\label{Lemma 2}
	Let us denote $\|h_1-h_2\|_Y=y$. Setting $k\in[3,\infty)$ and $p\in[k,\infty)$. We introduce $B(0,x)=\left\{f\in X| \|f\|_X\leq x \right\}$ as the close ball with radius $x$ in $X$. For a suitable $x$ and $d=d(x)$, the mapping $\mathcal{F}(.)$ satisfies the arguments as follows:
	\begin{enumerate}
		\item $\mathcal{F}: B(0,x)\rightarrow B(0,x)$ if $d<\tilde{\Lambda}_1(x)$, where
		\begin{eqnarray}
		\tilde{\Lambda}_1(x)=\frac{x\kappa^k\exp\left[-C_2(x^2+x^{k+1}+x^{p+1})\right]}{C_1(1+x^2+x^{k+1}+x^{p+1}+x^{p+3})}-\left(x^3+x^{k+2}+x^p+x^{p+2}+x^{p+4}\right),
		\end{eqnarray}
		\item there exists $\tilde{\Lambda}_2=\tilde{\Lambda}_2(x)\in[0,1)$, such that
		\begin{equation}
		\|\mathcal{F}(h_1)-\mathcal{F}(h_2)\|_Y \leq \tilde{\Lambda}_2\|h_1-h_2\|_Y.
		\end{equation}
	\end{enumerate}
\end{lemma}
\begin{proof}
From the definition of norm (\ref{norm x}), combining with (\ref{F Prisma}) and (\ref{G Prisma}), we have
\begin{align}\label{Estimate F}
\|\mathcal{F}\|_X\leq&\frac{C_1}{\kappa^k}\left(d+x^3+x^{k+2}+x^p+x^{p+2}+x^{p+4}\right)\left(1+x^2+x^{k+1}+x^{p+1}+x^{p+3}\right)\nonumber\\
&\times\exp\left[C_2(x^2+x^4+x^{p+1})\right].
\end{align}
Let us define
\begin{align}
\tilde{\Lambda}_1(x)=\frac{x\kappa^{k}\exp\left[-C_2(x^2+x^4+x^{p+1})\right]}{C_1(1+x^2+x^{k+1}+x^{p+1}+x^{p+3})}-\left(x^3+x^{k+2}+x^p+x^{p+2}+x^{p+4}\right).
\end{align}
We find that $\tilde{\Lambda}_1(0)=0, \tilde{\Lambda}_1'(0)>0,$ and $\tilde{\Lambda}_1(x)\rightarrow-\infty$ as $x\rightarrow\infty$. Thus, there exists $x_0\in(0,x_1)$ such that $\tilde{\Lambda}_1(x)$ is monotonically increasing on $[0,x_0]$. For every $x\in(0,x_0)$ we conclude that $\|\mathcal{F}\|_X\leq x$, and
\begin{align}
\mathcal{F}:B(0,x)\rightarrow B(0,x)~~~\mathrm{if}~~d<\tilde{\Lambda}_1(x).
\end{align}
This proves the first argument.

Now, we have the second argument. Setting $\Theta=\mathcal{F}(h_1)-\mathcal{F}(h_2)$, and the notation $g_l=g(h_l)$, $\mathcal{F}_l=\mathcal{F}(h_l)$, and $\mathcal{G}_l=\mathcal{G}(h_l)$, for $l=1,2$. We assume
\begin{align}
\max\left\{\|h_1\|_X,\|h_2\|_X\right\}<x.
\end{align}
Let us consider the equation (\ref{Dh}) for $h_1$ and $h_2$ in $X$ with $h_1(0,r_1)=h_2(0,r_2)$. Then, taking the differences between them such that we have
\begin{align}\label{Lipschitz Prisma}
&\frac{\partial\Theta}{\partial u}-\frac{\tilde{g}_1}{2}\frac{\partial\Theta}{\partial r}=\frac{1}{2}\left(\tilde{g}_1-\tilde{g}_2\right)\mathcal{G}_2+\frac{1}{2r}\left(g_1-\tilde{g}_1\right)(\Theta - \bar{h}_1 +\bar{h}_2)\nonumber\\
&+\frac{1}{2r}\left(g_1-\tilde{g}_1-g_2+\tilde{g}_2\right)(\mathcal{F}_2-\bar{h}_2)+4\pi r(g_1-g_2) \mathcal{F}_1V(\bar{h}_1)-4\pi r(g_1-g_2)\bar{h}_1 V(\bar{h}_1)\nonumber\\
&+4\pi r g_2\left(V(\bar{h}_1)-V(\bar{h}_2)\right)\mathcal{F}_2+4\pi r g_2\Theta V(\bar{h}_1)-4\pi r g_2(\bar{h}_1 - \bar{h}_2)V(\bar{h}_2)\nonumber\\
&-4\pi r g_2\bar{h}_1 \left(V(\bar{h}_1)-V(\bar{h}_2)\right) +\frac{r}{2}(g_1-g_2)\bar{h}_1\frac{\partial^2 V(\bar{h}_1)}{\partial \bar{h}_1} +\frac{r}{2}g_2(\bar{h}_1-\bar{h}_2)\frac{\partial^2 V(\bar{h}_1)}{\partial \bar{h}_1}\nonumber\\
&+\frac{r}{2}g_2\bar{h}_1\left(\frac{\partial^2 V(\bar{h}_1)}{\partial \bar{h}_1}-\frac{\partial^2 V(\bar{h}_2)}{\partial \bar{h}_2}\right).
\end{align}
As previously, we use the characteristic $\chi_l=\chi_l(u,r)$ defined by
\begin{align}
\frac{\mathrm{d}r}{\mathrm{d}u}=-\frac{\tilde{g}_1}{2}(\chi_1(u,r),u);~~~r(0)=r_0.
\end{align}
Thus, we write down the solution of (\ref{Lipschitz Prisma}) for $\Theta(u_1,r_1)$ as follows:
\begin{align}\label{Theta Prisma}
\Theta(u_1,r_1)=
\int_{0}^{u_1} \exp\left\{\int_{u}^{u_1} \left[\frac{1}{2r}\left(g_1-\tilde{g}_1\right)+4\pi r g_2V(\bar{h}_1)\right]_{\chi_1}\mathrm{d}u'\right\}[\tilde{\varphi}]_{\chi_1}\mathrm{d}u,
\end{align}
where we have denoted
\begin{align} \label{phi tilde}
\tilde{\varphi}=&\frac{1}{2}\left(\tilde{g}_1-\tilde{g}_2\right)\mathcal{G}_2-\frac{1}{2r}\left(g_1-\tilde{g}_1\right)( \bar{h}_1 -\bar{h}_2)+ \frac{1}{2r}\left(g_1-\tilde{g}_1+g_2-\tilde{g}_2\right)(\mathcal{F}_2-\bar{h}_2)\nonumber\\
&+4\pi r(g_1-g_2) \mathcal{F}_1V(\bar{h}_1)-4\pi r(g_1-g_2)\bar{h}_1 V(\bar{h}_1)+4\pi r g_2\left(V(\bar{h}_1)-V(\bar{h}_2)\right)\mathcal{F}_2\nonumber\\
&- 4\pi r g_2(\bar{h}_1 - \bar{h}_2)V(\bar{h}_2)-4\pi r g_2\bar{h}_1\left(V(\bar{h}_1)-V(\bar{h}_2)\right)\nonumber\\
&+\frac{r}{2}(g_1-g_2)\bar{h}_1\frac{\partial^2 V(\bar{h}_1)}{\partial \bar{h}_1^2} +\frac{r}{2}g_2(\bar{h}_1-\bar{h}_2)\frac{\partial^2 V(\bar{h}_1)}{\partial \bar{h}_1^2}+\frac{r}{2}g_2\bar{h}_1\left(\frac{\partial^2 V(\bar{h}_1)}{\partial \bar{h}_1^2}-\frac{\partial^2 V(\bar{h}_2)}{\partial \bar{h}_2^2}\right).
\end{align}

Let us denote $\|h_1-h_2\|_Y=y$. From the definition of mean value, we have
\begin{align}\label{hbar1-hbar2}
|\bar{h}_1 - \bar{h}_2|\leq \frac{1}{r}\int_{0}^{r}|h_1-h_2|\mathrm{d}s\leq\frac{y}{(k-2)}\frac{r^{k-3}}{(1+u)^{k-2}(1+r+u)^{k-2}}.
\end{align}
Thus,
\begin{align}\label{h1-h2}
|h_1-h_2-(\bar{h}_1-\bar{h}_2)|\leq\frac{Cy}{(1+u)^{k-2}(1+r+u)}.
\end{align}
Using (\ref{h-h}) and (\ref{h1-h2}), we obtain
\begin{align}\label{h-h square}
\left||h_1-\bar{h}_1|^2-|h_2-\bar{h}_2|^2\right|\leq \frac{Cxyr}{(1+u)^{2(k-2)}(1+r+u)^k}.
\end{align}
Then, we estimate
\begin{align}
|g_1-g_2|\leq4\pi \int_{r}^{\infty}\frac{1}{s}\left||h_1-\bar{h}_1|^2-|h_2-\bar{h}_2|^2\right|\mathrm{d}s\leq  \frac{Cxy}{(1+u)^{2(k-2)}(1+r+u)^{k-1}}.
\end{align}
From the above estimate, we obtain
\begin{align}\label{D1 1}
|\bar{g}_1-\bar{g}_2|\leq\frac{1}{r}\int_{0}^{r}|g_1-g_2|\mathrm{d}s\leq\frac{Cxy}{(1+u)^{3(k-2)}(1+r+u)}.
\end{align}
Now, we define
\begin{align}
\bar{h}_1^{p+1} - \bar{h}_2^{p+1}=(\bar{h}_1-\bar{h}_2)\int_{0}^{1}\left(t\bar{h}_1+(1-t)\bar{h}_2\right)\left|t\bar{h}_1+(1-t)\bar{h}_2\right|^{p-1}\mathrm{d}t,
\end{align}
such that
\begin{align}\label{hp1-hp2}
\left||\bar{h}_1|^{p+1}-|\bar{h}_2|^{p+1}\right|\leq\frac{2^px^py}{(k-2)^2(1+u)^{(k-2)(p+1)}(1+r+u)^{p+1}}.
\end{align}
From the above estimate, we obtain
\begin{align}\label{D1 2}
\left|\frac{8\pi}{r}\int_{0}^{r}gs^2\left(V(\bar{h}_1)-V(\bar{h}_2)\right)\mathrm{d}s\right|\leq& \frac{8\pi}{r}\int_{0}^{r}s^2\frac{K_02^px^py}{(k-2)^2(1+u)^{(k-2)(p+1)}(1+r+u)^{p+1}}\mathrm{d}s\nonumber\\
=& \frac{K_02^{p+3}\pi x^p y r^{k-2}}{(k-2)^2(1+u)^{(k-2)(p+2)}(1+r+u)^k}.
\end{align}
Combining (\ref{D1 1}) and (\ref{D1 2}), we have
\begin{align}\label{gtilde1-gtilde2}
|\tilde{g}_1-\tilde{g}_2|\leq&|\bar{g}_1-\bar{g}_2|+\left|\frac{8\pi}{r}\int_{0}^{r}gs^2\left(V(\bar{h}_1)-V(\bar{h}_2)\right)\mathrm{d}s\right|\nonumber\\
\leq&\frac{Cxy}{(1+u)^{3(k-2)}(1+r+u)}+\frac{K_02^{p+3}\pi x^p y r^{k-2}}{(k-2)^2(1+u)^{(k-2)(p+2)}(1+r+u)^k}\nonumber\\
\leq&\frac{Cy(x+x^p)}{(1+u)^{3(k-2)}(1+r+u)}.
\end{align}
From (\ref{h-h square}), we obtain
\begin{align}
\left|g_1-g_2-(\bar{g}_1-\bar{g}_2) \right|\leq&\frac{1}{r}\int_{0}^{r}\int_{r'}^{r}\left|\frac{\partial}{\partial r}(g_1-g_2)\right|\mathrm{d}s\mathrm{d}r'\nonumber\\
\leq&\frac{4\pi}{r}\int_{0}^{r}\int_{r'}^{r}\frac{1}{s}\left||h_1-\bar{h}_1|^2-|h_2-\bar{h}_2|^2\right|\mathrm{d}s\mathrm{d}r'\nonumber\\
\leq& \frac{Cxyr}{(1+u)^{2k-3}(1+r+u)^{k-1}}.
\end{align}
Thus,
\begin{align}\label{g1-g2}
\frac{1}{2r}\left|g_1-\tilde{g}_1-(g_2-\tilde{g}_2)\right|\leq&\frac{1}{2r}\left[\left|g_1-g_2-(\bar{g}_1-\bar{g}_2)\right|+\left|\frac{8\pi}{r}\int_{0}^{r}gs^2\left(V(\bar{h}_1)-V(\bar{h}_2)\right)\mathrm{d}s\right|\right]\nonumber\\
\leq& \frac{1}{2r}\left[ \frac{Cxyr}{(1+u)^{2k-3}(1+r+u)^{k-1}}+\frac{K_02^{p+3}\pi x^p y r^{k-2}}{(k-2)^2(1+u)^{(k-2)(p+2)}(1+r+u)^k}\right]\nonumber\\
\leq&\frac{C(x+x^p)y}{(1+u)^{2k-3}(1+r+u)^{k-1}}.
\end{align}
Combining (\ref{estimasi Potensial}), (\ref{F Prisma}), (\ref{G Prisma}), (\ref{hbar1-hbar2}), (\ref{gtilde1-gtilde2}), and (\ref{g1-g2}) we obtain (see Appendix \ref{Appendix3})
\begin{align}\label{varphi tilde}
|\tilde{\varphi}|\leq\frac{Cy\left[\alpha(x)+\beta(x)+\gamma(x)+\sigma(x)+x^2+x^{p-1}+x^{p+1}+x^{p+3}\right]}{(1+u)^{3(k-2)}(1+r+u)^2},
\end{align}
with
\begin{align}
\alpha(x)=&C(x+x^p)\left(d+x^3+x^{k+2}+x^p+x^{p+2}+x^{p+4}\right)(1+x^2+x^{k+1}+x^{p+1}+x^{p+3})\nonumber\\
&\times\exp\left[C(x^2+x^4+x^{p+1})\right]\label{alpha}\\
\beta(x)=&C(x+x^p)(d+x^3+x^p+x^{p+2})\exp\left[C(x^2+x^{p+1})\right]\label{beta}\\
\gamma(x)=&C x^p(d+x^3+x^p+x^{p+2})\exp[C(x^2+x^{p+1})]\label{gamma}\\
\sigma(x)=&Cx^{p+2}(d+x^3+x^p+x^{p+2})\exp[C(x^2+x^{p+1})]\label{sigma}.
\end{align}

We use the integral formulation such that
\begin{align}
\int_{u}^{u_1}\frac{1}{2r}\left|g_1-\tilde{g}_1\right|+4\pi r g_2|V(\bar{h}_1)|\mathrm{d}u\leq C(x^2+x^{p+1}).
\end{align}
Combining (\ref{int g2}) and (\ref{varphi tilde}), we write the estimate (\ref{Theta Prisma}) as follows:
\begin{align}
\left|\Theta(u_1,r_1)\right|\leq& \int_{0}^{u_1} \exp\left\{\int_{u}^{u_1} \left[\frac{1}{2r}\left|g_1-\tilde{g}_1\right|+4\pi r g_2V(\bar{h}_1)\right]_{\chi_1}\mathrm{d}u'\right\}\left|[\tilde{\varphi}]_{\chi_1}\right|\mathrm{d}u\nonumber\\
\leq&Cy\left[\alpha(x)+\beta(x)+\gamma(x)+\sigma(x)+x^2+x^{p-1}+x^{p+1}+x^{p+3}\right]\exp[C(x^2+x^{p+1})]\nonumber\\
&\times\int_{0}^{u_1}\frac{\mathrm{d}u}{(1+u)^{3(k-2)}(1+r+u)^2}.
\end{align}
Again, using the integral formulation (\ref{integral}), we obtain
\begin{align}\label{estimate Theta}
\left|\Theta(u_1,r_1)\right|\leq\frac{Cy\left[\alpha(x)+\beta(x)+\gamma(x)+\sigma(x)+x^2+x^{p-1}+x^{p+1}+x^{p+3}\right]\exp[C(x^2+x^{p+1})]}{\kappa^{k-1}(1+r_1+u_1)^{k-1}}.
\end{align}

From the definition of (\ref{norm y}), we have
\begin{align}\label{Lipschitz condition}
\|\Theta\|_Y \leq\tilde{\Lambda}_2y
\end{align}
with
\begin{align}\label{Lipschitz constant}
\tilde{\Lambda}_2(x)=\frac{C_3}{\kappa^{k-1}}\left[\alpha(x)+\beta(x)+\gamma(x)+\sigma(x)+x^2+x^{p-1}+x^{p+1}+x^{p+3}\right]\exp[C_4(x^2+x^{p+1})],
\end{align}
where $\alpha(x),\beta(x),\gamma(x),$ and $\sigma(x)$ are defined in (\ref{alpha})-(\ref{sigma}) respectively.

We find that $\tilde{\Lambda}_2(0)=0$ and $\tilde{\Lambda}'_2(0)>0$. Then, $\tilde{\Lambda}_2$ is monotonically increasing on $\tilde{x}_0\in \mathbb{R}^+$. There exists $x_2\in \mathbb{R}^+$ such that $\tilde{\Lambda}_2(x)<1$ for all $x$ in $(0,x_2]$. 
Hence, the mapping $h\mapsto\mathcal{F}(h)$ is contraction in $Y$ for $\|h\|_X\leq x_2$. This completes the proof.	
\end{proof}

\section{Global Existence of Classical Solution}
This section is devoted to proving that $h$ is a global classical solution of the equation (\ref{Dh}).
\begin{theorem} \label{Theorem 1}
	Given the initial data $h(0,r)\in C^1[0,\infty)$. Let $X$ and $Y$ be the function spaces defined by (\ref{space X}) and (\ref{space Y}) respectively. There exists a global classical solution of Equation (\ref{Dh}) for $k\in[3,\infty)$ and $p\in[k,\infty)$ such that
	\begin{eqnarray}
	h(u,r)\in C^1 ([0,\infty)\times [0,\infty)). 
	\end{eqnarray}
\end{theorem}
\begin{proof}
In Lemma \ref{Lemma 2}, we showed that $h\mapsto\mathcal{F}(h)$ is a contraction mapping in $Y$. Since $Y$ contains $X$, it is guaranteed that $h\mapsto\mathcal{F}(h)$ is also a contraction mapping in $X$. Using Banach's fixed theorem, there exists a unique fixed point $h\in X$ such that $\mathcal{F}(h)=h$. 

Now, we will show that $\left|\frac{\partial h}{\partial u}\right|$ is also bounded. From (\ref{Dh}), we have (see Appendix \ref{Appendix4})
\begin{eqnarray}\label{estimate dhdu}
\left|\frac{\partial h}{\partial u}\right|&\leq& \frac{\tilde{g}}{2}\left|\frac{\partial h}{\partial r}\right| + \frac{1}{2r}\left|g-\tilde{g}\right|\left|h\right| +\frac{1}{2r}\left|g-\tilde{g}\right|\left|\bar{h}\right|\nonumber\\
&&+4\pi gr\left|h\right||V(\bar{h})|+4\pi gr|\bar{h}||V(\bar{h})|+\frac{gr}{2}\left|\frac{\partial V(\bar{h})}{\partial \bar{h}}\right|\nonumber\\
&\leq&\frac{C(K_1(x)+K_2(x)+K_3+x^3+x^p+x^{p+2})}{(1+r+u)^2}.
\end{eqnarray}
This provides  that $\left|\frac{\partial h}{\partial u}\right|$ is bounded. The proof is finished.	
\end{proof}

\section{Appendix}

\subsection{The \{00\} and \{01\} components of the Einstein equations}
\label{Appendix Einstein}

We write the $\{00\}$ component of the Einstein equations (\ref{Einstein}) as follows
\begin{align} \label{00}
R_{00}=\frac{1}{2}g_{00}R + 8\pi T_{00}.
\end{align}
Combining (\ref{R00}), (\ref{Ricci scalar}), and (\ref{Energy-Momentum Tensor}), we write down Eq.(\ref{00}) such that
\begin{align}\label{Komponen 00}
&e^{2(F-G)}\left[-\frac{2}{r}\partial_1 G+\frac{1}{r^2}\right]+8\pi\left[\frac{1}{2}e^{2(F-G)}\left(\frac{\partial\phi}{\partial r}\right)^2 - e^{2F}V(\phi)+\left(\frac{\partial\phi}{\partial u}\right)^2-e^{F-G}\frac{\partial\phi}{\partial u}\frac{\partial\phi}{\partial r}\right]\nonumber\\
&=\frac{2}{r}e^{F-G}\partial_0 G.
\end{align}
The $\{01\}$ component of the Einstein equations (\ref{Einstein}) yields
\begin{align} \label{01}
R_{01}=\frac{1}{2}g_{01}R + 8\pi T_{01}.
\end{align}
As previously, from (\ref{R01}), (\ref{Ricci scalar}), and (\ref{Energy-Momentum Tensor}), we write down Eq.(\ref{01}) as follows
\begin{align}\label{Komponen 01}
&e^{2(F-G)}\left[-\frac{2}{r}\partial_1G+\frac{1}{r^2}\right]+8\pi\left[\frac{1}{2}e^{2(F-G)}\left(\frac{\partial\phi}{\partial r}\right)^2 - e^{2F}V(\phi)\right]=0.
\end{align}
Again, since there is no term containing the second derivative with respect to $u$ and $r$ on $F$ and $G$, this shows that $F$ and $G$ are no longer dynamical. Theorem \ref{Theorem 1} ensures that the metric functions $F$ and $G$ do globally exist.

\subsection{Estimate for (\ref{f prisma})}\label{Appendix1}
We write the estimate of (\ref{f prisma}) as follows:
	\begin{align}
	|f|\leq&\left|\frac{1}{2r}\left(g-\tilde{g}\right)\bar{h}\right| + \left|4\pi g rV(\bar{h})\bar{h}\right|+\left|\frac{gr}{2}\frac{\partial V(\bar{h})}{\partial \bar{h}}\right|\nonumber\\
	=& A_1+A_2+A_3.
	\end{align}

	\begin{enumerate}
		\item Estimate for $A_1$
		
		First, we estimate
		\begin{align}
		|\bar{h}|\leq \frac{1}{r}\int_{0}^{r}|h(u,s)|\mathrm{d}s	\leq \frac{x}{(k-2)}\frac{r^{k-3}}{(1+u)^{k-2}(1+r+u)^{k-2}}.
		\end{align}
		Then, we obtain
		\begin{align}
		|h(u,r)-h(u,r')|\leq&\int_{r'}^{r}\left|\frac{\partial h}{\partial s}(u,s)\right|~\mathrm{d}s\leq x\int_{r'}^{r}\frac{1}{(1+s+u)^k}~\mathrm{d}s\nonumber\\
		\leq&\frac{x}{k-1}\left[(1+r'+u)^{-k+1}-(1+r+u)^{-k+1}\right].
		\end{align}
		From the above estimate, we get
		\begin{align}
		|(h-\bar{h})(u,r)|\leq \frac{1}{r}\int_{0}^{r}|h(u,r)-h(u,r')|~\mathrm{d}r'\leq\frac{xr}{(k-1)(1+u)^{k-2}(1+r+u)^{k-1}},
		\end{align}
		and
		\begin{align}
		\int_{0}^{\infty}\frac{|h-\bar{h}|^2}{r}~\mathrm{d}r\leq&\frac{x^2 }{(k-1)^2(1+u)^{2(k-2)}}\int_{0}^{\infty}\frac{s}{(1+s+u)^{2(k-1)}}~\mathrm{d}s\nonumber\\
		=& \frac{x^2 }{(k-1)^2(12-14k+4k^2)(1+u)^{4(k-2)}}.
		\end{align}
		From (\ref{g2}), we get
		\begin{align}
		g(u,0)=\exp \left[-\int_{0}^{\infty}4\pi \frac{(h-\bar{h})^2}{s}\mathrm{d}s\right]\geq \exp \left[-\frac{4\pi x^2 }{(k-1)^2(12-14k+4k^2)}\right].
		\end{align}
		Thus,
		\begin{align}
		\bar{g}(u,0)=&\frac{1}{r}\int_{0}^{r}g(u,0)~\mathrm{d}r\geq\exp \left[-\frac{4\pi x^2 }{(k-1)^2(12-14k+4k^2)}\right].
		\end{align}
		Now, we estimate
		\begin{align}
		&|g(u,r)-g(u,r')|\leq\int_{r'}^{r}\left|\frac{\partial g}{\partial s}(u,s) \right|\mathrm{d}s\leq 4\pi \int_{r'}^{r}\frac{|h-\bar{h}|^2}{s}\mathrm{d}s\nonumber\\
		&=\frac{\pi x^2}{3k(k-1)^2(k-2)^2} \left[\frac{(1+u)^{4-2k}}{r'^2}+\frac{4k(1+r'+u)^{1-2k}}{1-2k}+\frac{4k^2(1+r'+u)^{1-2k}}{-1+2k}\right.\nonumber\\
		&+\frac{4ku(1+r'+u)^{1-2k}}{1-2k}+\frac{4k^2u(1+r'+u)^{1-2k}}{(-1+2k)}+\frac{(-1+k)(1+r'+u)^{1-2k}(1+(-1+2k)r'+u)}{-1+2k}\nonumber\\
		&-\frac{(1+u)^{4-2k}}{r^2}-\frac{4k(1+r+u)^{1-2k}}{1-2k}-\frac{4k^2(1+r+u)^{1-2k}}{-1+2k}-\frac{4ku(1+r+u)^{1-2k}}{1-2k}\nonumber\\
		&\left.-\frac{4k^2u(1+r+u)^{1-2k}}{(-1+2k)}-\frac{(-1+k)(1+r+u)^{1-2k}(1+(-1+2k)r+u)}{-1+2k}\right],
		\end{align}
		and
		\begin{align}
		|(g-\bar{g})(u,r)|\leq& \frac{1}{r}\int_{0}^{r}|g(u,r)-g(u,r')|\mathrm{d}r'\nonumber\\
		=& \frac{\pi x^2 r^2}{3k(k-1)^2(k-2)^2(1+u)^{-3+2k}(1+r+u)^3}.
		\end{align}
		We also estimate
		\begin{align}
		\int_{0}^{r}gs^2V(\bar{h})\mathrm{d}s\leq K_0\int_{0}^{r}gs^2|\bar{h}|^{p+1}\mathrm{d}s\leq\frac{K_0x^{p+1}r^3}{3(1+u)^{4(k-2)}(1+r+u)^3}.
		\end{align}
		Thus, we obtain
		\begin{align}
		\left|g-\tilde{g}\right|\leq|g-\bar{g}| + \frac{8\pi}{r}\int_{0}^{r}gs^2V(\bar{h})\mathrm{d}s\leq\frac{C(x^2+x^{p+1})r^2}{(1+u)^{-3+2k}(1+r+u)^3}.
		\end{align}
	
		Estimates (\ref{hbar}) and (\ref{g-gtilde}) yields
		\begin{align}\label{A1}
		A_1=\left|\frac{(g-\tilde{g})}{2r}\bar{h}\right|\leq\frac{C(x^3+x^{p+2})}{(1+u)^{3k-5}(1+r+u)^3}.
		\end{align}
		\item Estimate for $A_2$
		
		Using (\ref{hbar}), we obtain the estimate (\ref{A2}) as follows:
		\begin{align}\label{A2}
		A_2=\left|4\pi grV(\bar{h})\bar{h}\right|\leq \frac{Cx^{p+2}}{(1+u)^{5(k-2)}(1+r+u)^4}.
		\end{align}
		
		\item Estimate for $A_3$\\
		From (\ref{hbar}), we obtain
		\begin{align}\label{A3}
		A_3=\left|\frac{gr}{2}\frac{\partial V(\bar{h})}{\partial \bar{h}}\right|\leq\frac{Cx^p}{(1+u)^{3(k-2)}(1+r+u)^2}.
		\end{align}
	\end{enumerate}
	Combining (\ref{A1})-(\ref{A3}) yields
	\begin{align}
	|f|\leq\frac{C(x^3+x^p+x^{p+2})}{(1+u)^{3(k-2)}(1+r+u)^2}.
	\end{align}
\subsection{Estimate for (\ref{f1 Prisma})}\label{Appendix2}
We write the estimate of (\ref{f1 Prisma}) as follows:
	\begin{align}
	|f_1|\leq& \left[\frac{1}{2r}\left|\frac{\partial (g-\tilde{g})}{\partial r}\right| +\frac{1}{2r^2}|g-\tilde{g}|+\left|\frac{\partial g}{\partial r}\right|4\pi r|V(\bar{h})|+4\pi g|V(\bar{h})|+4\pi g r\left|\frac{\partial V(\bar{h})}{\partial\bar{h}}\right| \left|\frac{\partial\bar{h}}{\partial r}\right|\right] \left(|\mathcal{F}|+|\bar{h}|\right)\nonumber\\
	&+\left[\frac{1}{2r} |g-\tilde{g}|+4\pi g r|V(\bar{h})|+\frac{gr}{2}\left|\frac{\partial^2 V(\bar{h})}{\partial\bar{h}^2}\right|\right]\left|\frac{\partial \bar{h}}{\partial r}\right|+\left[\left|\frac{\partial g}{\partial r}\right|\frac{r}{2}+\frac{g}{2}\right]\left|\frac{\partial V(\bar{h})}{\partial\bar{h}}\right|\nonumber\\
	=& (B_1+B_2+B_3+B_4+B_5)\left(\left|\mathcal{F}\right|+\left|\bar{h}\right|\right) + (B_6+B_7+B_8)\left|\frac{\partial \bar{h}}{\partial r}\right| + B_9.
	\end{align}
	\begin{enumerate}
		\item Estimate for $B_1$
		
		Combining (\ref{hbar}), (\ref{g-gbar}), and (\ref{estimate1}) yields
		\begin{align}
		\left|\frac{\partial\tilde{g}}{\partial r}\right|\leq& \frac{C x^2 r}{(1+u)^{-3+2k}(1+r+u)^3}+\frac{8\pi K_0x^{k+1}r}{3(1+u)^{4(k-2)}(1+r+u)^3}+\frac{8\pi K_0x^{p+1}r}{(1+u)^{4(k-2)}(1+r+u)^{4}}\nonumber\\
		\leq& \frac{C(x^2+x^{k+1}+x^{p+1})r}{(1+u)^{-3+2k}(1+r+u)^3}.
		\end{align}
		Then, using (\ref{g tilde}) and (\ref{h-h}), we obtain
		\begin{align}
		\left|\frac{\partial g}{\partial r}\right|\leq \frac{4\pi x^2r }{(k-1)^2(1+u)^{2(k-2)}(1+r+u)^{2(k-1)}}.
		\end{align}
		Thus,
		\begin{align}\label{B1}
		B_1=\frac{1}{2r}\frac{\partial |g-\tilde{g}|}{\partial r}\leq \frac{C(x^2+x^{k+1}+x^{p+1})}{(1+u)^{2(k-2)}(1+r+u)^3}.
		\end{align}
		\item Estimate for $B_2$\\
		From (\ref{g-gtilde}), we obtain
		\begin{align}
		B_2=\frac{1}{2r^2}|g-\tilde{g}|\leq\frac{C(x^2+x^{p+1})}{(1+u)^{-3+2k}(1+r+u)^3}.
		\end{align}
		\item Estimaste for $B_3$\\
		We define $\frac{\partial\bar{h}}{\partial r}=\frac{h-\bar{h}}{r}$. From (\ref{hbar}) and (\ref{h-h}), we obtain
		\begin{align}
		B_3 =4\pi r\left|\frac{\partial g}{\partial r}\right||V(\bar{h})|\leq \frac{C x^{p+3}}{(1+u)^{6(k-2)}(1+r+u)^{2k}}.
		\end{align}
		\item Estimate for $B_4$\\
		From (\ref{hbar}), we obtain
		\begin{align}
		B_4=4\pi g V(\bar{h})\leq \frac{C x^{p+1}}{(1+u)^{4(k-2)}(1+r+u)^4}.
		\end{align}
		\item Estimate for $B_5$\\
		From (\ref{hbar}) and (\ref{gr}), we obtain
		\begin{align}
		B_5=4\pi g r \left|\frac{\partial V(\bar{h})}{\partial \bar{h}}\right|\frac{|h-\bar{h}|}{r}\leq \frac{C x^{p+1}}{(1+u)^{4(k-2)}(1+r+u)^{k+1}}.
		\end{align}
		\item Estimate for $B_6$\\
		From (\ref{g-gtilde}), we obtain
		\begin{align}
		B_6=\frac{1}{2r} |g-\tilde{g}|\leq \frac{C(x^2+x^{p+1})r}{(1+u)^{-3+2k}(1+r+u)^3}.
		\end{align}
		\item Estimate for $B_7$\\
		From (\ref{hbar}), we obtain
		\begin{align}
		B_7=4\pi g r V(\bar{h})\leq \frac{C x^{p+1}r}{(1+u)^{4(k-2)}(1+r+u)^4}.
		\end{align}
		\item Estimate for $B_8$\\
		Again, from (\ref{hbar}) we have
		\begin{align}
		B_8 =\frac{gr}{2}\left|\frac{\partial^2V(\bar{h})}{\partial \bar{h}^2}\right|\leq \frac{C x^{p-1}r}{(1+u)^{2(k-2)}(1+r+u)^2}.
		\end{align}
		\item Estimate for $B_9$\\
		From (\ref{hbar}) and (\ref{gr}), we obtain
		\begin{align}\label{B9}
		B_9=\left[\left|\frac{\partial g}{\partial r}\right|\frac{r}{2}+\frac{g}{2}\right]\left|\frac{\partial V(\bar{h})}{\partial \bar{h}}\right|\leq\frac{C x^{p+2}r}{(1+u)^{5(k-2)}(1+r+u)^{2k}}.
		\end{align}
	\end{enumerate}
	
	Combining (\ref{B1})-(\ref{B9}), yields
	\begin{align}
	|f_1|\leq& \frac{C(x^2+x^{k+1}+x^{p+1}+x^{p+3})}{(1+u)^{2(k-2)}(1+r+u)^3}|\mathcal{F}|+\frac{C(x^3+x^{k+2}+x^{p+2}+x^{p+4})}{(1+u)^{3(k-2)}(1+r+u)^4}\nonumber\\
	&+\frac{C(x^3+x^{p}+x^{p+2})r}{(1+u)^{3(k-2)}(1+r+u)^{k+1}}.
	\end{align}
\subsection{Estimate for (\ref{phi tilde})}\label{Appendix3}
We write the estimate of (\ref{phi tilde}) as follows:
\begin{align}
|\tilde{\varphi}|\leq&\frac{1}{2}\left|\tilde{g}_1-\tilde{g}_2\right||\mathcal{G}_2|+\frac{1}{2r}\left|g_1-\tilde{g}_1\right|| \bar{h}_1 -\bar{h}_2|+ \frac{1}{2r}\left|g_1-\tilde{g}_1+g_2-\tilde{g}_2\right||\mathcal{F}_2-\bar{h}_2|\nonumber\\
&+4\pi r|g_1-g_2| |\mathcal{F}_1||V(\bar{h}_1)|+4\pi r|g_1-g_2||\bar{h}_1| |V(\bar{h}_1)|+4\pi r g_2\left|V(\bar{h}_1)-V(\bar{h}_2)\right||\mathcal{F}_2|\nonumber\\
& + 4\pi r g_2|\bar{h}_1 - \bar{h}_2||V(\bar{h}_2)|+4\pi r g_2|\bar{h}_1|\left|V(\bar{h}_1)-V(\bar{h}_2)\right|\nonumber\\
&+\frac{r}{2}|g_1-g_2||\bar{h}_1|\left|\frac{\partial^2 V(\bar{h}_1)}{\partial \bar{h}_1^2}\right| +\frac{r}{2}g_2|\bar{h}_1-\bar{h}_2|\left|\frac{\partial^2 V(\bar{h}_1)}{\partial \bar{h}_1^2}\right|+\frac{r}{2}g_2|\bar{h}_1|\left|\frac{\partial^2 V(\bar{h}_1)}{\partial \bar{h}_1^2}-\frac{\partial^2 V(\bar{h}_2)}{\partial \bar{h}_2^2}\right|\nonumber\\
:=& D_1+D_2+D_3+...+D_{10}+D_{11}+D_{12}.
\end{align}
\begin{enumerate}
	\item Estimate for $D_1$\\
	Let us denote $\|h_1-h_2\|_Y=y$. From the definition of mean value, we have
	\begin{align}
	|\bar{h}_1 - \bar{h}_2|\leq& \frac{1}{r}\int_{0}^{r}|h_1-h_2|\mathrm{d}s\nonumber\\
	\leq&\frac{1}{r}\int_{0}^{r}\frac{\|h_1-h_2\|_Y}{(1+s+u)^{k-1}}\mathrm{d}s\nonumber\\
	\leq&\frac{y}{(k-2)}\frac{r^{k-3}}{(1+u)^{k-2}(1+r+u)^{k-2}}.
	\end{align}
	Thus,
	\begin{align}
	|h_1-h_2-(\bar{h}_1-\bar{h}_2)|\leq |h_1-h_2|+|\bar{h}_1-\bar{h}_2|\leq\frac{Cy}{(1+u)^{k-2}(1+r+u)}.
	\end{align}
	In view of (\ref{h-h}) and (\ref{h1-h2}), we have
	\begin{align}
	\left||h_1-\bar{h}_1|^2-|h_2-\bar{h}_2|^2\right|\leq& \left|(h_1-h_2)-(\bar{h}_1-\bar{h}_2) \right|\left(|h_1-\bar{h}_1|+|h_2-\bar{h}_2|\right)\nonumber\\
	\leq& \frac{Cxyr}{(1+u)^{2(k-2)}(1+r+u)^k}.
	\end{align}
	Then, we estimate
	\begin{align}\label{g1-g2 estimate}
	|g_1-g_2|\leq4\pi \int_{r}^{\infty}\frac{1}{s}\left||h_1-\bar{h}_1|^2-|h_2-\bar{h}_2|^2\right|\mathrm{d}s\leq  \frac{Cxy}{(1+u)^{2(k-2)}(1+r+u)^{k-1}}.
	\end{align}
	From the above estimate, we obtain
	\begin{align}
	|\bar{g}_1-\bar{g}_2|\leq\frac{1}{r}\int_{0}^{r}|g_1-g_2|\mathrm{d}s\leq\frac{Cxy}{(1+u)^{3(k-2)}(1+r+u)}.
	\end{align}
	Let us define
	\begin{align}
	\bar{h}_1^{p+1} - \bar{h}_2^{p+1}=(\bar{h}_1-\bar{h}_2)\int_{0}^{1}\left(t\bar{h}_1+(1-t)\bar{h}_2\right)\left|t\bar{h}_1+(1-t)\bar{h}_2\right|^{p-1}\mathrm{d}t,
	\end{align}
	such that
	\begin{align}
	\left||\bar{h}_1|^{p+1}-|\bar{h}_2|^{p+1}\right|\leq& |\bar{h}_1-\bar{h}_2|\left(|\bar{h}_1|+|\bar{h}_2|\right)^p\nonumber\\
	=&\frac{2^px^py}{(k-2)^2(1+u)^{(k-2)(p+1)}(1+r+u)^{p+1}}.
	\end{align}
	From the above estimate, we obtain
	\begin{align}
	\left|\frac{8\pi}{r}\int_{0}^{r}gs^2\left(V(\bar{h}_1)-V(\bar{h}_2)\right)\mathrm{d}s\right|\leq& \frac{8\pi}{r}\int_{0}^{r}s^2\frac{K_02^px^py}{(k-2)^2(1+u)^{(k-2)(p+1)}(1+r+u)^{p+1}}\mathrm{d}s\nonumber\\
	=& \frac{K_02^{p+3}\pi x^p y r^{k-2}}{(k-2)^2(1+u)^{(k-2)(p+2)}(1+r+u)^k}.
	\end{align}
	Combining (\ref{D1 1}) and (\ref{D1 2}), we have
	\begin{align}
	|\tilde{g}_1-\tilde{g}_2|\leq&|\bar{g}_1-\bar{g}_2|+\left|\frac{8\pi}{r}\int_{0}^{r}gs^2\left(V(\bar{h}_1)-V(\bar{h}_2)\right)\mathrm{d}s\right|\nonumber\\
	\leq&\frac{Cxy}{(1+u)^{3(k-2)}(1+r+u)}+\frac{K_02^{p+3}\pi x^p y r^{k-2}}{(k-2)^2(1+u)^{(k-2)(p+2)}(1+r+u)^k}\nonumber\\
	\leq&\frac{Cy(x+x^p)}{(1+u)^{3(k-2)}(1+r+u)}.
	\end{align}
	In view of (\ref{G Prisma}), we have
	\begin{align}
	D_1=&\frac{1}{2}\left|\tilde{g}_1-\tilde{g}_2\right||\mathcal{G}_2|\nonumber\\
	\leq&\frac{Cy(x+x^p)}{(1+u)^{3(k-2)}(1+r+u)}\frac{C\left(d+x^3+x^{k+2}+x^p+x^{p+2}+x^{p+4}\right)}{\kappa^k(1+r_1+u_1)^k}\nonumber\\
	&\times(1+x^2+x^{k+1}+x^{p+1}+x^{p+3})\exp\left[C(x^2+x^4+x^{p+1})\right].
	\end{align}
	From (\ref{ru2}), we obtain
	\begin{align}
	D_1\leq \frac{Cy\alpha(x)}{(1+u)^{3(k-2)}(1+r+u)^{k+1}},
	\end{align}
	where we have denoted
	\begin{align}
	\alpha(x)=&C(x+x^p)\left(d+x^3+x^{k+2}+x^p+x^{p+2}+x^{p+4}\right)(1+x^2+x^{k+1}+x^{p+1}+x^{p+3})\nonumber\\
	&\times\exp\left[C(x^2+x^4+x^{p+1})\right]\nonumber.
	\end{align}
	\item Estimate for $D_2$\\
	From (\ref{g-gtilde}), we obtain
	\begin{align}
	D_2=&\leq\frac{1}{2r}|g_1-\tilde{g}_1||\bar{h}_1-\bar{h}_2|\nonumber\\
	\leq&\frac{1}{2r}\frac{C(x^2+x^{p+1})r^2}{(1+u)^{-3+2k}(1+r+u)^3}\frac{y}{(k-2)}\frac{r^{k-3}}{(1+u)^{k-2}(1+r+u)^{k-2}}\nonumber\\
	\leq& \frac{C(x^2+x^{p+1})y}{(1+u)^{3k-5}(1+r+u)^{k}}.
	\end{align}
	\item Estimate for $D_3$\\
	From (\ref{h-h square}), we obtain
	\begin{align}
	\left|g_1-g_2-(\bar{g}_1-\bar{g}_2) \right|\leq&\frac{1}{r}\int_{0}^{r}\int_{r'}^{r}\left|\frac{\partial}{\partial r}(g_1-g_2)\right|\mathrm{d}s\mathrm{d}r'\nonumber\\
	\leq&\frac{4\pi}{r}\int_{0}^{r}\int_{r'}^{r}\frac{1}{s}\left||h_1-\bar{h}_1|^2-|h_2-\bar{h}_2|^2\right|\mathrm{d}s\mathrm{d}r'\nonumber\\
	\leq& \frac{Cxyr}{(1+u)^{2k-3}(1+r+u)^{k-1}}.
	\end{align}
	Thus,
	\begin{align}
	\frac{1}{2r}\left|g_1-\tilde{g}_1-(g_2-\tilde{g}_2)\right|\leq&\frac{1}{2r}\left[\left|g_1-g_2-(\bar{g}_1-\bar{g}_2)\right|+\left|\frac{8\pi}{r}\int_{0}^{r}gs^2\left(V(\bar{h}_1)-V(\bar{h}_2)\right)\mathrm{d}s\right|\right]\nonumber\\
	\leq& \frac{1}{2r}\left[ \frac{Cxyr}{(1+u)^{2k-3}(1+r+u)^{k-1}}+\frac{K_02^{p+3}\pi x^p y r^{k-2}}{(k-2)^2(1+u)^{(k-2)(p+2)}(1+r+u)^k}\right]\nonumber\\
	\leq&\frac{C(x+x^p)y}{(1+u)^{2k-3}(1+r+u)^{k-1}}.
	\end{align}
	In view of (\ref{F Prisma}), we have
	\begin{align}
	D_3=&\leq\frac{1}{2r}\left|g_1-\tilde{g}_1-(g_2-\tilde{g}_2)\right||\mathcal{F}_2|\nonumber\\
	\leq& \frac{C(x+x^p)y}{(1+u)^{2k-3}(1+r+u)^{k-1}}\frac{C(d+x^3+x^p+x^{p+2})\exp\left[C(x^2+x^{p+1})\right]}{\kappa^{k-1}(1+r_1+u_1)^{k-1}}.
	\end{align}
	From (\ref{ru2}), we obtain
	\begin{align}
	D_3\leq\frac{C y \beta(x)}{(1+u)^{2k-3}(1+r+u)^{2(k-1)}},
	\end{align}
	where we have denoted
	\begin{align}
	\beta(x)=C(x+x^p)(d+x^3+x^p+x^{p+2})\exp\left[C(x^2+x^{p+1})\right]\nonumber.
	\end{align}
	\item Estimate for $D_4$\\
	From (\ref{hp1-hp2}) and (\ref{F Prisma}), we obtain
	\begin{align}
	D_4=&4\pi r g_2\left(V(\bar{h}_1)-V(\bar{h}_2)\right)|\mathcal{F}_2|\nonumber\\
	\leq& \frac{K_0 r 2^px^py}{(1+u)^{p+1}(1+r+u)^{p+1}}\frac{C(d+x^3+x^p+x^{p+2})\exp\left[C(x^2+x^{p+1})\right]}{k^2(1+r_1+u_1)^2}.
	\end{align}
	From (\ref{ru2}), we obtain
	\begin{align}
	D_4\leq\frac{Cy \gamma(x)}{(1+u)^{4(k-2)}(1+r+u)^{k+2}},
	\end{align}
	where we have denoted
	\begin{align}
	\gamma(x)=C x^p(d+x^3+x^p+x^{p+2})\exp[C(x^2+x^{p+1})]\nonumber.
	\end{align}
	\item Estimate for $D_5$\\
	In view of (\ref{g1-g2}) and (\ref{hbar}), we get
	\begin{align}
	D_5=&\frac{1}{2r}\left|g_1-\tilde{g}_1-(g_2-\tilde{g}_2)\right||\bar{h}_2|\nonumber\\
	\leq&\frac{C(x+x^p)y}{(1+u)^{2k-3}(1+r+u)^{k-1}}\frac{x}{(k-2)(1+u)^{k-2}(1+r+u)}\nonumber\\
	\leq&\frac{C(x^2+x^{p+1})y}{(1+u)^{3k-5}(1+r+u)^k}.
	\end{align}
	\item Estimate for $D_6$\\
	From (\ref{hbar}), (\ref{g1-g2 estimate}), and (\ref{F Prisma}), we have
	\begin{align}
	D_6=&4\pi r|g_1-g_2||V(\bar{h}_1)| |\mathcal{F}_1|\nonumber\\
	\leq&Cr\frac{xy}{(1+u)^{2(k-2)}(1+r+u)^{k-1}}\frac{K_0x^{p+1}}{(k-2)^{p+1}(1+u)^{(k-2)(p+1)}(1+r+u)^{p+1}}\nonumber\\
	&\times\frac{C(d+x^3+x^p+x^{p+2})\exp\left[C(x^2+x^{p+1})\right]}{\kappa^{k-1}(1+r_1+u_1)^{k-1}}.
	\end{align}
	Then, from (\ref{ru2}) we obtain
	\begin{align}
	D_6\leq\frac{Cy\sigma(x)}{(1+u)^{6(k-2)}(1+r+u)^{2k+1}},
	\end{align}
	where we have denoted
	\begin{align}
	\sigma(x)=Cx^{p+2}(d+x^3+x^p+x^{p+2})\exp[C(x^2+x^{p+1})]\nonumber.
	\end{align}
	\item Estimate for $D_7$\\
	From (\ref{hbar}) and (\ref{g1-g2 estimate}), we obtain
	\begin{align}
	D_7 =&4\pi r|g_1-g_2||\bar{h}_1| |V(\bar{h}_1)|\nonumber\\
	\leq& Cr \frac{xy}{(1+u)^{2(k-2)}(1+r+u)^{k-1}}\frac{x}{(k-2)(1+u)^{k-2}(1+r+u)}\nonumber\\
	&\times\frac{K_0x^{p+1}}{(k-2)^{p+1}(1+u)^{(k-2)(p+1)}(1+r+u)^{p+1}}\nonumber\\
	\leq&\frac{Cyx^{p+3}}{(1+u)^{7(k-2)}(1+r+u)^{3+k}}.
	\end{align}
	\item Estimate for $D_8$\\
	Estimate (\ref{hbar}) and (\ref{hbar1-hbar2}) yields
	\begin{align}
	D_8 =&4\pi r g_2|\bar{h}_1 - \bar{h}_2||V(\bar{h}_2)|\nonumber\\
	\leq& Cr \frac{yr^{k-3}}{(k-2)(1+u)^{k-2}(1+r+u)^{k-2}}\frac{K_0x^{p+1}}{(k-2)^{p+1}(1+u)^{(k-2)(p+1)}(1+r+u)^{p+1}}\nonumber\\
	\leq&\frac{Cyx^{p+1}}{(1+u)^{5(k-2)}(1+r+u)^4}.
	\end{align}
	\item Estimate for $D_9$\\
	From (\ref{hbar}) and (\ref{hp1-hp2}), we have
	\begin{align}
	D_9 =&4\pi r g_2\left|V(\bar{h}_1)-V(\bar{h}_2)\right||\bar{h}_1|\nonumber\\
	\leq& \frac{4\pi K_02^px^pyr}{(k-2)^2(1+u)^{(k-2)(p+1)}(1+r+u)^{p+1}} \frac{x}{(k-2)(1+u)^{k-2}(1+r+u)}\nonumber\\
	\leq&\frac{Cyx^{p+1}}{(1+u)^{5(k-2)}(1+r+u)^4}.
	\end{align}
	\item Estimate for $D_{10}$\\
	Again, from (\ref{hbar}) and (\ref{g1-g2 estimate}) we have
	\begin{align}
	D_{10}=&\frac{r}{2}|g_1-g_2||\bar{h}_1|\left|\frac{\partial^2 V(\bar{h}_1)}{\partial \bar{h}_1^2}\right|\nonumber\\
	\leq& Cr\frac{xy}{(1+u)^{2(k-2)}(1+r+u)^{k-1}}\frac{K_0x^p}{(k-2)^p(1+u)^{p(k-2)}(1+r+u)^p}\nonumber\\
	\leq&\frac{Cyx^{p+1}}{(1+u)^{5(k-2)}(1+r+u)^{k+1}}.
	\end{align}
	\item Estimate for $D_{11}$\\
	In view of (\ref{hbar}) and (\ref{hbar1-hbar2}), we get
	\begin{align}
	D_{11} =&\frac{r}{2}g_2|\bar{h}_1-\bar{h}_2|\left|\frac{\partial^2 V(\bar{h}_2)}{\partial \bar{h}_2^2}\right|\nonumber\\
	\leq& Cr \frac{yr^{k-3}}{(k-2)(1+u)^{k-2}(1+r+u)^{k-2}}\frac{K_0x^{p-1}}{(k-2)^{p-1}(1+u)^{(k-2)(p-1)}(1+r+u)^{p-1}}\nonumber\\
	\leq&\frac{Cyx^{p-1}}{(1+u)^{3(k-2)}(1+r+u)^2}.
	\end{align}
	\item Estimate for $D_{12}$\\
	From (\ref{hp1-hp2}), we have
	\begin{align}
	\left|\bar{h}_1^{p-1}-\bar{h}_2^{p-1}\right|\leq& |\bar{h}_1-\bar{h}_2|\left(|\bar{h}_1|+|\bar{h}_2|\right)^{p-2}\nonumber\\
	\leq& \frac{K_0yx^{p-2}}{(1+u)^{(k-2)(p-1)}(1+r+u)^{p-1}}.
	\end{align}
	From the above estimate, combined with (\ref{hbar}) we have
	\begin{align}
	D_{12}=&\frac{r}{2}g_2\left|\frac{\partial^2 V(\bar{h}_1)}{\partial \bar{h}_1^2}-\frac{\partial^2 V(\bar{h}_2)}{\partial \bar{h}_2^2}\right||\bar{h}_1|\leq\frac{r}{2}g_2\left||\bar{h}_1|^{p-1}-|\bar{h}_2|^{p-1}\right||\bar{h}_1|\nonumber\\
	\leq& C r \frac{K_0yx^{p-2}}{(1+u)^{(k-2)(p-1)}(1+r+u)^{p-1}} \frac{x}{(k-2)(1+u)^{k-2}(1+r+u)}\nonumber\\
	\leq&\frac{Cyx^{p-1}}{(1+u)^{3(k-2)}(1+r+u)^2}.
	\end{align}
\end{enumerate}
Thus, we have
\begin{align}
|\tilde{\varphi}|\leq\frac{Cy\left[\alpha(x)+\beta(x)+\gamma(x)+\sigma(x)+x^2+x^{p-1}+x^{p+1}+x^{p+3}\right]}{(1+u)^{3(k-2)}(1+r+u)^2},
\end{align}
where $\alpha(x),\beta(x),\gamma(x),$ and $\sigma(x)$ are defined in (\ref{alpha})-(\ref{sigma}) respectively.

\subsection{Estimate for (\ref{estimate dhdu})}\label{Appendix4}
We write the estimate of (\ref{estimate dhdu}) as follows:
\begin{eqnarray}
\left|\frac{\partial h}{\partial u}\right|&\leq& \frac{\tilde{g}}{2}\left|\frac{\partial h}{\partial r}\right| + \frac{1}{2r}\left|g-\tilde{g}\right|\left|h\right| +\frac{1}{2r}\left|g-\tilde{g}\right|\left|\bar{h}\right|\nonumber\\
&&+4\pi gr\left|h\right||V(\bar{h})|+4\pi gr|\bar{h}||V(\bar{h})|+\frac{gr}{2}\left|\frac{\partial V(\bar{h})}{\partial \bar{h}}\right|\nonumber\\
&:=& H_1+H_2+H_3+H_4+H_5+H_6.
\end{eqnarray}
\begin{enumerate}
	\item Estimate for $H_1$\\
	From (\ref{G Prisma}), we have
	\begin{eqnarray}\label{estimate 1}
	\frac{\tilde{g}}{2}\left|\frac{\partial h}{\partial r}\right|\leq\frac{K_1(x)}{\kappa^{k}(1+r_1+u_1)^{k}}\leq \frac{K_1(x)}{(1+r+u)^k},
	\end{eqnarray}
	where we have denoted
	\begin{eqnarray}
	K_1(x)=&C(d+x^3+x^{k+2}+x^p+x^{p+2}+x^{p+4})(1+x^2+x^{k+1}+x^{p+1}+x^{p+3})\nonumber\\
	&\times\exp\left[C(x^2+x^4+x^{p+1})\right].
	\end{eqnarray}
	\item Estimate for $H_2$\\
	The estimate (\ref{F Prisma}) yields
	\begin{align}
	\frac{1}{2r}\left|g-\tilde{g}\right|\left|h\right|\leq&\frac{1}{2r} \frac{C(x^2+x^{p+1})r^2}{(1+u)^{-3+2k}(1+r+u)^3}\frac{C(d+x^3+x^p+x^{p+2})\exp\left[C(x^2+x^{p+1})\right]}{\kappa^{k-1}(1+r_1+u_1)^{k-1}}\nonumber\\
	\leq&\frac{K_2(x)}{(1+u)^{-3+2k}(1+r+u)^{k+1}},
	\end{align}
	where
	\begin{align}
	K_2(x)=C(x^2+x^{p+1})(d+x^3+x^p+x^{p+2})\exp\left[C(x^2+x^{p+1})\right].
	\end{align}
	\item Estimate for $H_3$
	\begin{align}
	\frac{1}{2r}\left|g-\tilde{g}\right|\left|\bar{h}\right|\leq& \frac{1}{r}\frac{C(x^2+x^{p+1})r^2}{(1+u)^{-3+2k}(1+r+u)^3}\frac{xr^{k-3}}{(k-2)(1+u)^{k-2}(1+r+u)^{k-2}}\nonumber\\
	\leq&\frac{C(x^3+x^{p+2})}{(1+u)^{3k-5}(1+r+u)^k}.
	\end{align}
	\item Estimate for $H_4$\\
	Combining (\ref{hbar}) and (\ref{F Prisma}), we obtain
	\begin{align}
	4\pi g r |h||V(\bar{h})|\leq& Cr\frac{C(d+x^3+x^p+x^{p+2})\exp\left[C(x^2+x^{p+1})\right]K_0x^{p+1}}{\kappa^{k-1}(1+r_1+u_1)^{k-1}(k-2)^{p+1}(1+u)^{(k-2)(p+1)}(1+r+u)^{p+1}}\nonumber\\
	\leq&\frac{K_3(x)}{(1+u)^{4(k-2)}(1+r+u)^{k+2}},
	\end{align}
	where
	\begin{eqnarray}
	K_3(x)=Cx^{p+1}(d+x^3+x^p+x^{p+2})\exp\left[C(x^2+x^{p+1})\right].
	\end{eqnarray}
	\item Estimate for $H_5$
	\begin{align}
	4\pi g r |\bar{h}||V(\bar{h})|\leq& Cr \frac{xr^{k-3}}{(k-2)(1+u)^{k-2}(1+r+u)^{k-2}}\frac{K_0x^{p+1}}{(k-2)^{p+1}(1+u)^{(k-2)(p+1)}(1+r+u)^{p+1}}\nonumber\\
	\leq&\frac{Cx^{p+2}}{(1+u)^{(k-2)(p+2)}(1+r+u)^{k+1}}.
	\end{align}
	\item Estimate for $H_6$\\
	The estimate (\ref{estimasi Potensial}) yields
	\begin{eqnarray}\label{estimate 6}
	\frac{gr}{2}\left|\frac{\partial V(\bar{h})}{\partial\bar{h}}\right|\leq \frac{Cx^p}{(1+u)^{p(k-2)}(1+r+u)^{p-1}}.
	\end{eqnarray}
\end{enumerate}

Combining (\ref{estimate 1}) - (\ref{estimate 6}), we write the estimate (\ref{estimate dhdu}) as follows
\begin{eqnarray}
\left|\frac{\partial h}{\partial u}\right|\leq\frac{C(K_1(x)+K_2(x)+K_3+x^3+x^p+x^{p+2})}{(1+r+u)^2}.
\end{eqnarray}

\section*{Acknowledgments}
The work of this research is supported by PPMI FMIPA ITB 2022, PPMI KK ITB 2022, and GTA 50 ITB. B.E.G. would like to acknowledge the support from the ICTP through the Associate's Programme (2017-2022). E.S.F. also would like to acknowledge the support from GTA Research Group ITB and from BRIN through the Research Assistant Programme 2022.

\section*{Data Availability}
This manuscript has no associated data.

\section*{Conflict of Interest}
No conflict of interest in this paper.

\appendix

\end{document}